\documentclass[sigconf]{acmart}
\usepackage[ruled]{algorithm2e} 

\usepackage{graphicx}
\usepackage{geometry}
\usepackage[utf8]{inputenc} 
\usepackage[T1]{fontenc}    
\usepackage{amsfonts}       
\usepackage{nicefrac}       
\usepackage{microtype}      
\usepackage{subfig}
\acmConference[Preprint]{Preprint}{\today}{}

\begin{document}

\settopmatter{printacmref=false} 
\renewcommand\footnotetextcopyrightpermission[1]{} 
\pagestyle{plain} 
\title{Cost Per Action Constrained Auctions}
\author{Benjamin Heymann}

\def \b {{{b^i}}}
\def \bb {{{\hat{b}}}}
\newcommand{\E}{\mathbb{E}}
\newcommand {\R}{\mathbb{R}}
\newcommand {\T}{{\rm T}}
\newcommand{\dd}{{\rm d}}
\newcommand{\rev}{{\mathcal{R}}}
\newcommand{\cost}{{\mathcal{C}}}
\newtheorem{remark}{Remark}
\begin{abstract}
A standard result from auction theory is that bidding truthfully in a second price auction is a weakly dominant strategy. 
The result, however, does not apply in the presence of Cost Per Action (CPA) constraints.
Such constraints exist, for instance, in digital advertising, as
some buyer may try to maximize the total number of clicks while keeping the empirical Cost Per Click (CPC) below a threshold. 
More generally the CPA constraint implies that the buyer has a maximal average cost per unit of value in mind.

We discuss how such constraints change some traditional results from auction theory.

Following the usual textbook narrative on auction theory, we focus specifically on the symmetric setting,  
We formalize the notion of CPA constrained auctions and derive a Nash equilibrium for second price auctions. We then extend this result to combinations of first and second price auctions. 
Further, we expose a revenue equivalence property and show that the seller's revenue-maximizing reserve price is zero.

In practice, CPA-constrained buyers may  target an  empirical  CPA on a given time horizon, as the auction is repeated many times. Thus his bidding behavior depends on past realization. 
We show that the resulting buyer dynamic optimization problem can be formalized with  stochastic control tools and solved numerically with available solvers.

\end{abstract}
\settopmatter{printfolios=true}
\maketitle


\section{Introduction} 

What should you bid in a second price ad auction for a display with a known click-through rate (CTR), for a given cost-per-click (CPC)?  
The most probable (and possibly incorrect) answer is "CPC times CTR". However, the right answer is "What do you mean by cost-per-click?". 
Indeed, if the CPC is the maximal amount the buyer is ready to pay for each click, the first answer is correct. But the story is different when  CPC means the maximal average  cost per click.

In this article, we challenge the current modeling approach for ad auctions. We argue that when some advertisers' business constraints apply, the expected outcome of the auction may depart from the traditional literature.  

Advertising is a major source of revenue for Internet publishers, and as such, is financing a large part of the Internet.  About  200 billion USD spent in 2017 on digital advertising \cite{Kafka2017}.
Banners for display advertising are usually bought through a high frequency one unit auction mechanism called RTB (Real-Time Bidding) by or on behalf of advertisers.

When a user reaches a publisher page, it triggers a call to an RTB platform. The RTB platform then calls potential buyers, who have a few milliseconds to answer with a bid request.
This results in an allocation of the display to a bidder, in exchange for a payment. The allocation and the payment are defined by the rules of the auction, which depends on the RTB platforms. 
 The  winner of the auction can show a banner to the internet user. 
Yet, for many end buyers, what really matters is not the display itself, but what will result from it: clicks, conversions, etc...
This is why empirical CPA measures such as average cost per click (CPC) or average cost per order (CPO) are of paramount importance.  

The literature on display advertising auctions has been growing over the last decade, pushed by the emergence of  Internet giants whose business models rely on digital advertising. 
One track of research takes the seller's point of view  and focuses on how to build "good" auctions, and relies on mechanism design theory
\cite{balseiro2015repeated,golrezaei2017boosted}, while the dual  track brings the buyer perspective  under scrutiny, and focuses on the design of bidding strategies
\cite{zhang2014multi,agarwal2014budget,diemert2017attribution,fernandez2016optimal,cai2017real}.
The reader can refer to \cite{roughgarden2016twenty} and \cite{krishna2009auction} for an introduction to auction theory.

Usually, in performance markets, the advertiser has a target Cost Per Action (CPA) and/or a budget in mind. 
 We will focus the analysis on the CPA constraint. 
We will assume that the budget constraint is absent or not binding. This topic of budget constraints has already been addressed in the literature \cite{balseiro2015repeated,agarwal2014budget,fernandez2016optimal}. The CPA is measured in term of the average quantity of money spent per action, the action is a click, a visit, a conversion... 
For instance, if 
(a) we are facing a competition uniformly distributed between $0$ and $1$,
(b) the CTR of every display is $0.5$,
(c) the auction is second price and
(d) we are ready to pay 1 USD on average for a click,
then if we bid $0.5$, we will win half of the time. For every display won, we get on average half a click, and we pay $2\int_0^{0.5} t dt = 1/4$, thus our expected CPC is $1/(4*0.5)=0.5<1$. We should increase the bid to raise the empirical CPC.

One reason for using CPA constraints instead of budgets constraints in performance markets is that if the campaign performs well, there is no reason to shut it down at the middle of the month because of a strict budget limit.  Those constraints, implemented by algorithms, can be modified by comparing the performances of different channels. 

Auctions with ROI constraints are a hot topic \cite{Auerbach2008,Wilkens2016,Golrezaei2018}.
In particular \cite{Golrezaei2018}  provide some empirical evidence that some buyers are ROI constrained and propose an optimal auction design. 
Our work departs from them because we use a different notion of ROI:  in the definition of CPA, we do not subtract the cost to the revenue in the numerator. 
The reason to do so is that end buyers reason in « cost per click » or « cost per order ». The numerator is not expressed in a unit of money and does not include the buyer payment to the seller. 
Also, our focus is more buyer than seller oriented.

In this work, we focus on the buyers perspective, but we also show on simple market instances that the CPA constraints may impact the seller design, or give rise to undesirable competitive behaviors.

This work brings several contributions to the table.
First, we introduce the buyer's CPA constrained optimization problem and exhibit a solution, 
then we find the symmetric equilibrium in this setting and compare its properties with standard results from the literature. In particular, no reserve price should be used.
Last, since  in practice, the CPA is computed over a time window of repeated auctions, we move the discussion to the dynamic settings for which we explain why one can expect adaptive bid multipliers to provide solutions close to the optimal for the buyer.

\section{Modeling Assumptions and Notations}

We start by introducing the  CPA-constrained bidding problem.

A seller is selling one item (a display opportunity in our context)  through an auction.
The item brings a value $v_i$ to  bidder $i\in{1..n}$. 
These values are independently distributed among bidders.
For example, in the context of the  ad auction, $v$ can be thought of as the expected number of clicks (or sales) the bidder would get after winning the auction. Thus, this value is not expressed in money but in unit of action.
We will see later how these values, expressed in actions relates to bids, expressed in dollars.
Let $F_i$ be the  cumulative  distribution and $f_i$ the  density distribution of $v_i$. For simplicity, we assume the support of these density distributions to be compact intervals.

When we take the viewpoint of  bidder $i$, 
we denote by $b^{-i}$ the  greatest   bid of the other bidders (the price to beat). We denote by $g^{i}$ and $G^{i}$  its density and  cumulative distributions.
Unless otherwise stated, we also assume the auction to follow a second price rule, with $n>1$ bidders competing for the item.

Bidders with CPA constraints compare:
\begin{itemize}
\item the expected value they get in the auction
\item the expected payment they incur. 
\end{itemize}

We will refer to the CPA with the letter $\T$ because it is our "targeted cost per something". 
We will denote by $b^i$ the bid of a given bidder $i$ of interest.

The bidder wants to maximize his expected value, subject to an \emph{ex ante} constraint in expectancy representing the targeted CPA.  The constraint is \emph{ex ante} because, in practice, the same bidding strategy is going to be applied repeatedly (or even simultaneously) on several similar auctions. 
If for a random variable $X$, we denote by $[X]\in\{0,1\}$ the binary random variable that takes the value $1$ when $X\geq 0$, and $0$ otherwise, then the expected value earned by the buyer bidding $\b$ is $E_{v,b^-}[b^-<\b(v)] v$.
Here $E_{v,b^-}$ is there expectation operator on the distribution of $v$ and $b^-$.
Similarly the bidder  expected cost is $\E_{v,b^-} [b^-<\b(v)]b^- $.

We can now express the bidder optimization problem.  The bidder is looking for a bid function $\b:[0,1]\rightarrow\R^+$ solution of
  \begin{eqnarray}
    \label{eq:bidderProblem}
   \max_{\b}\quad \E_{v,b^{-i}} [b^{-i}<\b(v)]v, 
  \end{eqnarray}
 subject to  $ \E_{v,b^{-i} }[b^{-i}<\b(v)]b^{-i} \leq \T^i \E_{v,b^{-i}}[b^{-i}<\b(v)] v$ (CPA).
Observe that if we remove the constraint, the buyer would buy all the opportunities no matter the cost. 
We pinpoint that problem \eqref{eq:bidderProblem} is not equivalent to a budget constraint problem: we would have had $ \E_{v,b^{-i}} [b^{-i}<\b(v)]b^{-i} \leq budget$
instead of the (CPA) constraint. 
Nor is it equivalent to the maximization of the yield $\E_{v,b^-}(T^iv - b^{-i})[b^{-i}<b^i(v)]$, as illustrated in the introduction.

We model the  interactions among buyers with  a game. Observe that this game has  constraints on the strategies profile. 
For a given set of competing bidding strategies, a best reply of the bidder is a solution of (\ref{eq:bidderProblem}).
A (constrained) Nash equilibrium  is a strategies profile with components that are best replies against the others . 

\textit{In the following, the superscript $i$ is  often omitted to help readability.}

\section{Bidding Behavior and Symmetric Equilibrium} 
The main result of this section is the derivation of a symmetric Nash equilibrium for an CPA constrained second price auction. We then generalize this result to linear combinations of first price and second price auctions. 
\subsection{Second Price}

Going back to the motivational example of the introduction, if a buyer bids $Tv$, then he would pay \emph{at most} $T$ per unit of action. 
Since he can increase his bid until  he pays $T$ per unit of action \emph{on average}, it is intuitive that he should bid higher than $Tv$ to maximize  the criterion. The next theorem formalizes this idea. 
\begin{theorem}[Best Reply]
\label{theorem:secondprice-best-reply}
For any bounded distribution of the price to beat $b^-$, there exists $\lambda^*\in \bar{\R}^+$ such that a solution of \eqref{eq:bidderProblem}  writes: 
\begin{equation}
\label{eq:best-reply}
  \b(v,\lambda^*) = (\nicefrac {1} {\lambda^*} +\T^i )v.
\end{equation}
\end{theorem}

\begin{lemma}[CPA monotony]
\label{lemma:roi-monotony-second-price}
If the bidder bids proportionally to the value i.e. if  there exists $\alpha$ such that $b(v) = \alpha v$, then the CPA is non-decreasing in $\alpha$.
\end{lemma}
\begin{proof}
Fix $u\in[0,1]$, $\alpha>0$, then the expected cost knowing that the value is $v$ and the auction is won is $\int_0^{\alpha v} t g(t) \dd t$, where $g$ is the probability density of the price to beat. This quantity is increasing in $\alpha$. Meanwhile the value is $v$.
We get the result by integration on $v$.
\end{proof}

If $v_1, v_2, \ldots v_n$ are $n$ independent draws from $F$, we denote by $v_i^{(n)}$ the average of the $i^{th}$ greater draw (the ith order statistic).  
\begin{theorem}[Equilibrium Bid]
\label{theorem:equilibrium-bid-second-price}
  The unique  symmetric constrained Nash equilibrium strategy is
\begin{equation}
\b^*(v) = \frac{ \T}{\gamma}   v,
\end{equation}
where 
\begin{equation}
\label{eq:gamma-def}
\gamma(F,n) =  n \left(\frac{v_1^{(n-1)}}{v_1^{(n)}} - \frac{n-1}{n}\right).
\end{equation}
\end{theorem}

\paragraph{Discussion}

Theorem \ref{theorem:secondprice-best-reply} is the answer to the introductory question. Despite its simplicity, it shows something that is probably overlooked: it may happen in a second price auction that the bidder's optimal bid depends on the competition. Another useful insight is that the bid is linear w.r.t. the value, which implies that simple bid multipliers could be used optimally. Observe that the result holds also for non-symmetric settings.

The proof relies on  the strong incentive compatibility of the second price auction in standard setting reinterpreted on  the Lagrangian of the optimization problem. 
The parameter $\lambda^*$ should be interpreted as the Lagrangian multiplier associated with the CPA constraint. 

Informally speaking, the harder the constraint, the bigger $\lambda$, the smaller the bid.
When $\lambda$ is set to $+\infty$, one  bids only for displays  that do not consume the constraint, while when $\lambda$ goes to zero, the bid diverges.

Lemma \ref{lemma:roi-monotony-second-price} is  useful  understand what is happening. Observe  that since the bid is linear w.r.t.  the value and the  objective increasing in the bid multiplier, the optimal bid  multiplier is the maximal admissible  one.

Theorem \ref{theorem:equilibrium-bid-second-price} is one of our main result. Observe that
 the \textbf{competition factor} $\gamma$ is only a function of the number of opponents and $F$. Since it is  smaller than one (see Appendix) the bid is  proportional to  $\T v$ with a factor greater than one.
 Compared with the standard setting, the seller's expected  revenue   is multiplied by   $\nicefrac{ \T}{\gamma} $.
The asymptotic value  of $\gamma$, as well as other natural questions concerning the equilibrium will be studied   in Section \ref{sec:properties-of-the-equilibrium}.
 Note that Theorem \ref{theorem:equilibrium-bid-second-price}   is a necessary and sufficient condition for a symmetric Nash equilibrium, therefore it is unique, but we cannot claim  that we have identified the only Nash Equilibrium  (see Section \ref{sec:properties-of-the-equilibrium}).

\subsection{Generalization}
\label{subseq:generalization}

We now generalize the argument to auctions which are convex combinations of first price and second price.
We motivate this extension by the existence  in the industry of auctions which mix together first and second pricing rules, such as soft-floor.

If we denote by $S(b,b^-)$ the payment rule of the auction when the buyer bids $b$ and the best competing bid is $b^-$, then we can make the following observations: 
(1) for any $s\geq 0$, $S(s b,sb^-) = sS(b,b^-)$,
(2)  $S(b,b) = b$, 
(3) $S(b,b^-)$ is the sum of a linear function of $b$ and a linear function of $b^-$.

Consider a standard auction with payment rule $S$ and symmetric buyers with i.i.d. values distributed according to F (no CPA constraints). Let $\bb$ be a symmetric equilibrium strategy in such setting.

\begin{theorem}
  \label{theorem:equilibirum_strategy-general}
The bid $\b^*(v) = \frac{ \T}{\gamma}  \bb(v)$ is a symmetric equilibrium strategy, of the CPA constrained auction with payment rule $S$.
\end{theorem}

\paragraph{Example}: 
Consider a first price auction with 2 buyer having uniformly distributed values in $[0,1]$.
It is known  \cite{krishna2009auction} that an equilibrium bid for a standard auction is $v \rightarrow 0.5 v$.
By \eqref{eq:gamma-power-case}, $\gamma = 0.5$.
Thus an equilibrium bid for the CPA constrained auction is 
$v \rightarrow T v$.
Indeed: in first price, we ensure the saturation of the CPA constraint by bidding $Tv$.

\section{Properties of the Equilibrium} 
\label{sec:properties-of-the-equilibrium}

Observe that the expected payment of a buyer is equal to his expected value times $\T$. Since the total welfare is fixed, the expected payment of every buyer is the same no matter the auction. 
We thus recover a standard result \cite{vickrey1962auctions} from auction theory, adapted to CPA constrained auctions. 
\begin{theorem}[Revenue equivalence]
All the auctions described in Section \ref{subseq:generalization} bring the same expected revenue to the seller.
\end{theorem}

\begin{proof}
Using the proof of Theorem 3.3, we see that 
constraint in \eqref{eq:bidderProblem} is binding at the  symmetric equilibrium. Therefore the payment of a buyer to the seller is equal to  $T\E v[v>v^{-}] =T\E \max_{i\in\{1..n\}}(v_{i})/n$. 
\end{proof}
On the other hand, the following observation is quite unusual. 

\begin{theorem}[Optimal Reserve Price]
\label{th:reserve-price}
With CPA constraints, the optimal reserve price in a symmetric setting is zero.
\end{theorem}
Explanation: a reserve price would decrease the social welfare, and thus decrease the expected payment because of the CPA constraint.
\begin{proof}
Same idea as before: 
In the presence of a reserve price,  the constraint in \eqref{eq:bidderProblem} tells us that 
the buyer  expected payment is smaller than 
$T\E v[v>v^{-i}] [v>r]$, which is  smaller than $T\E v[v>v^{-i}]$  for r>0. 
Using the proof of Theorem 3.3 we see that the constraint should be binding at the equilibrium, and thus
this last quantity corresponds to the expected payment when $r=0$.
Thus the reserve price should be set to zero.
\end{proof}

\begin{figure}%
    \centering
   \includegraphics[width=7cm]{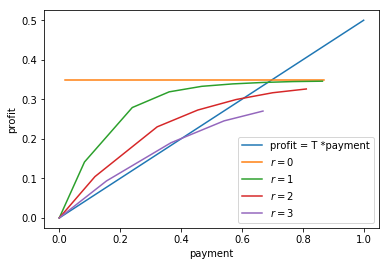} %
    \caption{Potential outcomes (buyers' profits and payments) for different values of the bid multiplier in a symmetric setting. Each curve corresponds to a different value of the reserve price. The equilibrium points are the intersections of those curves with the line profit = $T$*payment. \emph{We took $T = 0.4$, uniform distribution of values, and two bidders for the numerical simulation.}} %
    \label{fig:optinalReservePrice}%
\end{figure}

We display in Figure \ref{fig:optinalReservePrice} the results of a simulation that illustrates Theorem \ref{th:reserve-price}.
Observe that the proofs of the two last results rely on the symmetry of the setting (same CPA target $T$, same value distribution  $F$).
We argue that this should not be seen as a limitation of our results, because (a) Revenue Equivalence 
in the classical setting also requires symmetry, (b) 
Theorem \ref{th:reserve-price}  provides a striking  counterexample to the common belief that the reserve price increases the revenue of the seller.
We believe Theorem \ref{th:reserve-price} to be extendable to nonsymmetric setting, as long as one can guarantee that the CPA constraint is binding.

This completely departs from the idea that reserve prices should be used to increase the seller's revenue. Observe that in practice in the case of display advertising, the buyer may take some time to react to a change of empirical CPA. 
Consequently, measuring the seller's long term expected revenue uplift is a technical and business challenge.
Moreover, the reserve price, by decreasing the welfare,  may on the long run trigger an increase of the target CPA, since buyers have to do an arbitrage  between volume and CPA. 

Last but not least,  we  do not claim that reserve prices should be set to zero for CPA constrained bidders in all situations.
A trivial counterexample to such claim  is  embodied by a setting with only one buyer and no competition.

\begin{theorem}[Convergence]
\label{convergence}

 For any  number of bidders $n$,  $\gamma(F,n)\leq1$. Moreover
if the value $v$ is bounded, then
\begin{equation}
\lim_{n\rightarrow \infty} \gamma(F,n) = 1
\end{equation}
\end{theorem}
\begin{proof}
Let us denote by $Y^{(k)}$ the maximum of $k$ random variables drawn with the distribution $F$.  One just need to observe that
\begin{equation}
\label{eq:new-gamma-def}
\gamma(F,n) = 1 -\frac{\E(v-Y^{n-1}|v>Y^{n-1})}{v_1^{(n)}}.
\end{equation}

\end{proof}
This can be interpreted in the following way: when there is a great number of bidders, the competition is such that the payment tends to become first price.

\begin{remark}
There may be  a nonsymmetric equilibrium.
\end{remark}
The basic idea is that one bidder can bid higher than necessary to force another bidder to leave the auction because his (CPA) constraint cannot be satisfied.
\begin{proof}

We  concentrate on exhibiting a counter-example.
Take $n=2$, $T=1$, $F$ is the uniform distribution over $[2,3]$, and denote by $(\alpha_i)_{i=1,2}$ the bid multipliers of the two buyers.
Set $\alpha_2 =6$. Then if $\alpha_1\leq 4$, buyer 1 does not win any auction, but the (CPA) constraint is satisfied.

On the other hand, if $\alpha_1 > 4$, then the cost vs value ratio is bigger than 4, and the (CPA) constraint is violated. 
We can check that if $\alpha_1 =0$ the (CPA) constraint is satisfied for $2$.
Therefore $(\alpha_1,\alpha_2) =(0,6)$ is an asymmetric Nash equilibrium.
\end{proof}

In addition to this, observe that the buyers may be tempted to bid other strategies than the linear best reply.
For example consider this  example: 
Take $n=2$, an exponential distribution with parameter 1, $T=1$. 
One buyer may increase its profit by bidding with an \emph{affine function}. 
Compare $(b_1(v_1),b_2(v_2)) =(2v_1 +1,2v_2) $ with $(b_1(v_1),b_2(v_2)) =(3v_1 ,3v_2) $. 

On a simulation with $10^8$ auctions, we get in the first case an empirical CPA of 
$0.92$, and a revenue of $0.84$ for bidder 1 (resp. $1.0$ and $0.61$ for bidder 2), while in the second case, we get
an empirical CPA  of
$1$, and a revenue of $0.75$ for bidder 1 (resp. $1.0$ and $0.75$ for bidder 2).
Those simple, informal examples indicate that the bidders may be tempted to bid aggressively to weaken the other bidders CPA.

\section{Dynamic Bidder Problem}
\label{sec:dynamic}
In practice, the buyer behaviors may differ from the static case: 
(1)  the buyer can adapt his bidding strategy to  the past events, 
(2) the linear constraint (CPA) does not reflect the buyer risk aversion, 
(3) the benefit the buyer gets from a won auction is stochastic (for example, if he is only interested in clicks or conversions), 
(4) from a business perspective,  there is a trade-off between CPA and volumes in the buyer mind. This trade-off can be expressed in different ways. 

We propose in this section a continuous time optimal control framework to express and study the buyer's problem in a dynamic fashion. The approach combines the advantages of a powerful and mathematically clean expressiveness with  theoretical insights and numerical tools. 

We refer the reader to \cite{pham2009continuous} for a reference on stochastic control, and to \cite{Falcone2013} for a presentation of the numerical methods involved. The User Guide \cite{bonnans:hal-01192610} provides a hands-on presentation  of the topic. 

We only model one individual buyer. Yet note that the optimal control formulation is the  first building block for  the  study of the full market dynamics (with mean field games\cite{Carmona2013,gueant2011} or differential games\cite{Isaacs1999}).
Our main discovery in the stochastic case  is that (a) the optimal bid is not linear in the value anymore,  but (b)  we can  still reasonably propose  a linear bid (at the cost of an approximation to be discussed thereafter).

\subsection{Deterministic  Dynamic Formulation}
We use the deterministic case to introduce some notations, 
and recover with optimal control tools some of the results we already derived with the static model. Indeed, when we neglect the stochastic aspect of the reward, the problem is very similar to the static one. 

We consider a continuous time model. The buyer receives a continuous flow of requests.  
The motivation to propose a continuous time model is that it makes the theoretical analysis easier \cite{pham2009continuous,Heymann2018}.

We denote by $\rev$ the instantaneous revenue and by $\cost$ the instantaneous cost: 

\begin{equation}
\rev(b) = \E v [b^{-}<b(v)],
\quad 
\mbox{and}
\end{equation}

\begin{equation}
\cost(b)  = \E b^{-} [b^{-}<b(v)],
\end{equation}
If we denote by $\tau$ the time horizon, the bidder is now maximizing with respect to $b$
\begin{eqnarray}
\label{eq:MDP}
\int_0^\tau \rev(b_t)\dd t,
\end{eqnarray}
subject to 
$\dot{X}_t = \T \rev(b_t) - \cost(b_t)$, 
$X_\tau\geq 0$ and
$X_0 = x_0$.

The state $X_t$ represents the (CPA) constraint is should be positive at $t=\tau$ for the (CPA) constraint to be satisfied.
We use the Pontryagin Maximum Principle to argue that a necessary condition for an optimal control is that there exists a multiplier $p(t)$ so that 
$b(t)$ maximizes $H = p(\T\rev - \cost) +\rev$. 
Moreover $\dot{p} = -\frac{\partial H}{\partial x }=0$. We thus recover the main result of the static case: the bid is linear in $v$. Moreover
in the absence of stochastic components, the multiplier is constant.
Observe that in practice, one could use the solution of  problem (\ref{eq:MDP}) to implement a Model Predictive Control (MPC), using $x_0$ as the current state and thus get an online bidding engine.

\subsection{Stochastic Formulation}

 The number of displays is much higher than the number of clicks/sales, therefore we neglect the  randomness of  the competition/price over the randomness that comes from the \emph{actions}.
 Because of the number of displays involved, we argue that by the law of large number,
 the uncertainty on the action outcome can be apprehended  a white  noise.
 We thus add an additional, stochastic  term in the revenue:
$  \T\sigma_i\dd W_t$, where $W_t$ is a Brownian motion. 
We get  that an optimal  bid  $b$ should maximize (see \cite{pham2009continuous} for a reference)
\begin{equation}
\E_{b^-,v}  [b >b^-]((T v - b^-)p +  v  + \T^2 \sigma^2 M ),
\end{equation}
for some $p$ and $M$.

Since the realization of the action follows  a binomial law, $\sigma \propto v(1-v)$. 
 Assuming $v<<1$, we can approximate it as $v$.
Therefore  every infinitesimal terms of the Hamiltonian becomes
$
((T v - b^-)p +  v  + v M \T^2)[b>b^-]
$
reproducing the previous reasoning we get 
\begin{equation}
b^* = v (\T+\frac{M\T^2+1}{p}).
\end{equation}

\paragraph{Conclusion:} Once again, the bid factor approach is justified!
Observe that this conclusion comes from an approximation, (therefore the bid is not  strictly speaking optimal for the initial problem), but by doing so, we have reduced the control set from $\R^+\rightarrow\R^+$ to $\R^+$. 
This reduction gives access to  a well understood approach to solve this kind of problem:  the continuous time equivalent of Dynamic Programming.

Observe that our problem definition is incomplete: we need to take into account the constraint on the final state. 
We restate this constraint with a penalty $K(X)$.
Using the dynamic programming principle, the buyer problem is equivalent to the following Hamilton-Jacobi-Bellman equation:
\begin{equation}
 V_t + \sup_{\alpha} \left(
 (\T \rev(\alpha) - \cost(\alpha)) V_x+
\frac{\T^2}{2} \rev(\alpha)  V_{xx} +
\rev(\alpha)
\right) =0 
\end{equation}
\begin{equation}
\mbox{ and }  \quad V_\tau = K,
\end{equation}
with 
\begin{equation}
\rev(\alpha) = \rev(v\rightarrow \alpha v)  ,
\quad 
\mbox{and}
\end{equation}
\begin{equation}
\cost(\alpha) = \cost(v\rightarrow \alpha v) 
\end{equation}

\subsection{Numerical Resolution}

 Our aim here is to  illustrate that HJB approaches provide powerful  numerical and theoretical tools to model the buyer's preferences. A quantitative comparison of the performance of the stochastic method is out of the scope of this article.

We solve an instance of the problem with the numerical  optimal stochastic control solver BocopHJB \cite{bonnans:hal-01192610}.
BocopHJB computes the Bellman function by solving the non-linear partial differential equation (HJB).
Observe that the state is one dimensional, which makes this kind of solver very efficient.

We take $G(b^-) = (b^-)^{a}$ on $[0,1]$
and $v$ uniformly distributed on $[0,1]$.
We get that $\rev(\alpha) = \frac{\alpha^{a}}{a+2}$ if $\alpha <1$ and 
$\rev(\alpha)  = \frac{\alpha^{a}}{(a+2)(\alpha^{a+2})} + \frac{1}{2}(1- \frac{1}{\alpha^2})$ otherwise, 
similarly, $\cost(\alpha) = \frac{a\alpha^{a+1}}{(a+1)(a+2)}$ for $\alpha \leq 1$ and 
$\cost(\alpha) = \frac{a}{(a+1)(a+2)\alpha}((a+2)\alpha -a -1)$ for $\alpha \geq 1$ (see Figure \ref{fig:R}).
We take $\T=0.8$ and a linear penalty for negative  $X_\tau$. The result of a  simulated trajectories is displayed in Figure \ref{fig:simulation2}(a). 

We see that on this specific instance of the stochastic path, the control is decreased at the end of the time horizon to adapt to the sudden spike of CPA. This is the kind of 
behavior we can expect from a stochastic controller. By contrast,  a constant multiplier may result in the kind of trajectories displayed on Figure \ref{fig:simulation2} (b).

\begin{figure}%
    \centering
    \subfloat[a=1]{{\includegraphics[width=7cm]{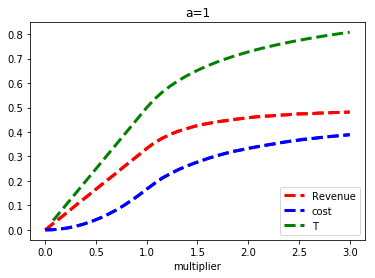} }}%

    \subfloat[a=5]{{\includegraphics[width=7cm]{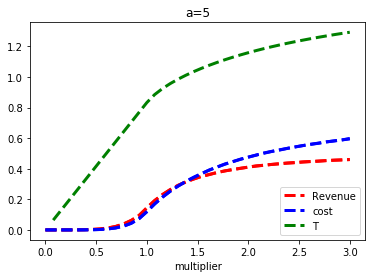} }}%
    \caption{$\rev$ and $\cost$ and empirical $\T$ for two values of $a$ }%
    \label{fig:R}%
\end{figure}

\begin{figure}%
    \centering
    \subfloat[stochastic controller]{{\includegraphics[width=7cm]{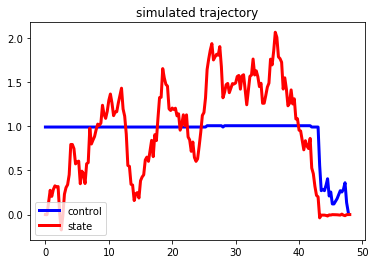} }}%

        \subfloat[deterministic controller]{{\includegraphics[width=7cm]{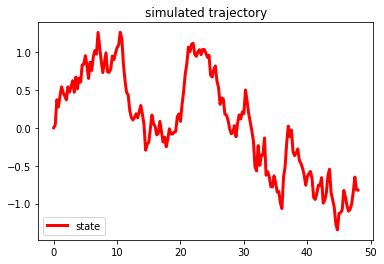}} }%
    \caption{Simulation result}%
    \label{fig:simulation2}%
\end{figure}

\section{Conclusion} 

We have formalized the existence of CPA constraints in the buyers bidding problem. 
These constraints are of paramount importance in performance markets, yet they tend to be neglected in the current literature. 
We have seen how standard results translate into this context, sometimes with surprising insights.
We also envisioned the dynamic, stochastic  setting by combining a first-order approximation with the Hamilton-Jacobi-Bellman approach. This allowed us in particular to derive a numerical method.

This work raises questions that  deserve more extended study. 
In particular, can we generalize the first three theorems to other constraints, auction rules and settings? 
Can the existence of the aggressive behaviors pointed out at the end of Section 5 pose  a problem to the market?
One may also want to put under scrutiny the harder, but more realistic case of correlated values. 
The dynamic case  also  deserves a specific study.

\begin{acks}

I would like to thank J\'er\'emie Mary  for his helpful comments.

\end{acks}

\bibliographystyle{plainnat}

\newpage 
\appendix
\section{Examples} 

\subsection{Power Law Case: derivation of the competition factor}
 We take $F(v) = v^a$ over $[0,1]$ with $a>0$. We get that $v^{(n)_1} = \frac{an}{an+1}$. Therefore we can write
$\frac{v_1^{(n-1)}}{v_1^{(n)}} =\frac{n-1}{n}\frac{an+1}{a(n-1)+1} $, 
which leads to 
 \begin{equation}
 \label{eq:gamma-power-case}
 \gamma(a,n) =(n-1) \left( \frac{an+1}{a(n-1)+1}-1  \right). 
 \end{equation}
 We observe that $\gamma $ is increasing in $n$ and $a$. It converges to 1 as  $n$ (or $a$) goes to infinity.
A plot of  $\gamma$ is displayed in Figure \ref{fig:example}(a). 
\paragraph{Remark:}
If we consider the 2 bidders case, and denote by $\alpha>0$ and $\beta>0$ their respective bid multipliers, then the payment and expected value of the first bidder are respectively equal to %
\begin{equation}
\pi = \beta (\frac{\alpha}{\beta})^{a+2} \frac{1}{(a+2)(2a+3)}, \mbox{ and }
V =  \frac{(\frac{\alpha}{\beta})^{a+1}}{(a+1)(2a+3)}. 
\end{equation}
We see that the CPA of the first bidder does not depend on the second bidder strategy. Therefore, we have in this case something similar to a strategy dominance. 
\subsection{Numerical  estimation of the competition factor}

The trick is  to remark that $\gamma$ is equal to the ratio of cost and value  when $\T =1$. We derive the following algorithm to estimate $\gamma$  for  different $F$ and $n$:
\begin{enumerate}
\item Choose $n$, $F$ and a number of samples $N$
\item $Cost  = 0$; $Value  = 0$;   
\item Generate $v\sim F$ and $v^- \sim F^n$
\item If $v>v^-$ then: $Cost+=v^-$ and $Value+=v$
\item Repeat from 3. (N-1) times
\item $\gamma \approx $\nicefrac{Cost}{Value}
\end{enumerate}
The results are displayed in Figure \ref{fig:example}.
For the log-normal  distribution, the results are similar to those of the power distribution: monotony and convergence to 1.
It seems that the bigger the variance, the bigger $\gamma$. Without providing a proof, we see that we can get an insightful intuition from Equation \ref{eq:new-gamma-def} of the next section.

The exponential distribution provides more exotic results: $\gamma$ seems to be only a function of $n$ (i.e. independent of the scale parameter).
This can be easily  proved analytically using the formula $\E(max(v_1,v_2))= \frac{1}{\lambda_1}+ \frac{1}{\lambda_1}- \frac{1}{\lambda_1+\lambda_2}$ and a recursion on $v_1^{(n)}$. 
It is however still increasing in $n$, and seems to converge to 1 as $n$ goes to infinity.

\begin{figure}%
    \centering
    \subfloat[Power Distribution]{{\includegraphics[width=7cm]{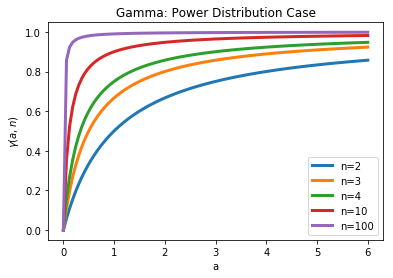} }}%
    //
    \subfloat[Exponential Distribution]{{\includegraphics[width=7cm]{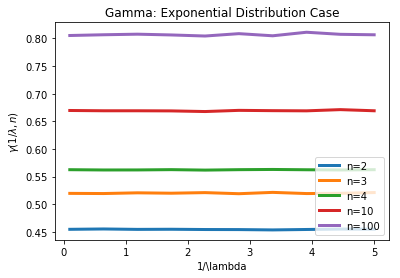} }}%
    \qquad
        \subfloat[LogNormal Distribution]{{\includegraphics[width=7cm]{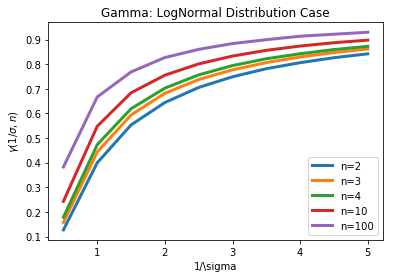} }}%
    \caption{The scaling factor as a function of the distribution parameter, for different numbers of bidders}%
    \label{fig:example}%
\end{figure}

\section{Proof of Theorem 3.1} 
Set $b(\alpha,v) = \alpha v$ for $\alpha>0$.
If  $b(\alpha,v)$ is admissible  for any $\alpha$, then we can  set $\lambda = 0$ and we are done (edge case).
Else, there exists $\alpha^+$ such that  
$\gamma(\alpha^+)>0$, where $\gamma(\alpha) :=\E_{v}[b^{-i}< b(\alpha,v)](b^{-i} - \T^i v) $

Moreover:
\begin{enumerate}
    \item since $\alpha \rightarrow\int_{b^{-i}/\alpha}^1 (b^{-i} - \T^i v)f^{i}(v)\dd v
$ is continuous for any $b^{-i}>0$, so is
$\gamma$,
\item for $\alpha$ smaller than $ T^{i}$, $\gamma(\alpha)\leq 0$.
\end{enumerate}   
So  there exists $\alpha^*$ satisfying
$\gamma(\alpha^*) = 0$.

If we set $\lambda^*$ so that $\frac{1}{\lambda^*} + T^{i} =\alpha^*$,
then by strategic dominance of the truthful strategy in  a second price auction, 
$b(\alpha^*,v)$ solves the unconstrained optimization problem of maximizing
\begin{equation}
\label{eq:intergand}
  \E_{b^-,v}[b^- <b](v - \lambda^* (b^- -  T^{i}v)).  
\end{equation}

We can also assert that  $b(\alpha^*,.)$ is also an optimal  solution  for the constrained optimization problem \eqref{eq:bidderProblem}.
Why? First observe that $b(\alpha^*,.)$ is admissible since $\gamma(\alpha^*) = 0$,
second, let $\hat{b}$ be an optimal solution to \eqref{eq:bidderProblem}, then, if we write $k(b) =  \E_{b^-,v}[b^- <b]v$, we have that  
\begin{equation}
k(b(\alpha^*)) - \lambda^* \gamma(\alpha^*)\geq k(\hat{b}) - \lambda^* \E_{v}[b^{-i}< \hat{b}](b^{-i} - \T^i v)
\end{equation}
thus since  $\gamma(\alpha^*) = 0$ and $\E_{v}[b^{-i}< \hat{b}](b^{-i} - \T^i v)\leq 0$
\begin{equation}
k(b(\alpha^*))\geq k(\hat{b}).
\end{equation}

\section{Proof of Theorem 3.3}
To derive the first order condition, we follow the same path as in the proof of Theorem \ref{theorem:equilibirum_strategy-general}.
We thus only  discuss the uniqueness. 
We start with a small technical consideration: observe that the bids can be modified on a zero measure set without impacting the profit. 
So when we say unique, we ,in fact, mean essentially unique. 

We denote by $b$ a symmetric equilibrium bidding strategy.
It induces a distribution $G$ of the price to beat. 
It is clear that we cannot be in the edge case  envisioned in the previous proof. 
So take $\alpha^*$  as introduced in the previous proof. 
Observe that to be a best reply $b$ has to maximize \eqref{eq:intergand}.
Assume there exists $v_0$ such that  $\alpha v \neq b(v)$ on a neighborhood $V$ of $v_0$.
Since $g(b(v))>0$ for $v\in V$, $\alpha v $ is strictly dominating $b(v)$ in the optimization of \eqref{eq:intergand}, which is a contradiction.

\section{Proof of Theorem 3.4} 

First observe that, the Lagrangian of the bidder optimization problems writes: \[ \lambda \E_{v,b^-}   \left((\nicefrac{1}{\lambda}+ \T ) v - S(b(v),b^{-})\right)[b>b^-],\] thus we can look for a point-wise maximizer. 
Denote by $(\hat{b},\hat{b}^-)$ the equilibrium bid of the same auction without any CPA constraints.
Observe that if  the price to beat is distributed like $(\nicefrac{1}{\lambda}+ \T ) \hat{b}^-(v^-)$ (for a given $\lambda>0$), then:
\begin{eqnarray*}
b^*(v) = \mbox{argmax}_b \E_{v^-}   (\nicefrac{1}{\lambda}+ \T ) v - S(b,b^{-}(v^-)) \\
 =(\nicefrac{1}{\lambda}+ \T )   \mbox{argmax}_b \E_{v^-}   (\nicefrac{1}{\lambda}+ \T ) v - S((\nicefrac{1}{\lambda}+ \T ) b,(\nicefrac{1}{\lambda}+ \T ) \hat{b}^-(v^-)) \\
 =(\nicefrac{1}{\lambda}+ \T )   \mbox{argmax}_b \E_{v^-}   v - S( b,\hat{b}^-(v^-))\\
 = (\nicefrac{1}{\lambda}+ \T ) \hat{b}(v).
\end{eqnarray*}
Therefore, for a given $\lambda$, if the competition uses the proposed linear bid, one should also bid linearly to maximize the Lagrangian. 
We denote by $G$ the price to beat distribution and by $\kappa$  the degree of "second priceness": $S(b,b^-) = ((1-\kappa)b + \kappa b^-)[b>b^-]$.
Observe that the first term of the integrand in the Lagrangian definition is equal to   $\E_v b_\lambda G(b_\lambda)$.
So, let us deal with the second term. Since

\begin{eqnarray*}
 \E_{v,v^-}  S(b_{\lambda}(v), b^-(v^-)) [v>v^-]
= \E_v \int_0^{b_\lambda(v)} S(b_\lambda(v),t)g(t) dt\\
= \E_v S(b_\lambda(v),b_\lambda(v)) G(b)  - \int_0^{b_\lambda(v)}\kappa G(t) dt \\
= \E_v b_\lambda G(b_\lambda) - \int_0^{b_\lambda(v)}\kappa G(t) dt, 
\end{eqnarray*}
we get that $$L =\lambda \kappa \E_v \int_0^{b_\lambda(v)} G(t) dt. $$

The first order condition on $\lambda$ writes
\[
0 = \E_{v,b^-}   \left((\nicefrac{1}{\lambda}+ \T ) v - S(b(v),b^{-})\right)[b>b^-] -\nicefrac{\kappa}{\lambda}  \E_v  G(b_\lambda(v)) v \dd v.
\]
Then using successively the homogeneity of $S$, a simplification by $\kappa$, and the de facto definition of $\gamma$ as $\nicefrac{\lambda\T}{1+\lambda \T}$, we get  
\[\gamma = \frac{\E v^- [v>v^-]}{\E v [v>v^-]},
\]
which we simplify using the following formula on the order statistics (see \cite{krishna2009auction}): $\E Y^{(n)}_2 =\E n Y_1^{(n-1)} - (n-1)Y_1^{(n)}$.

\end{document}